\documentclass[11pt]{article}
\usepackage[margin=1.1in]{geometry}
\usepackage{amsmath,amsthm}
\usepackage{color}
\usepackage{graphicx}
\usepackage{tikz}
\usepackage{bm}
\usepackage{tkz-berge}
\usepackage{amssymb, framed, hyperref, wrapfig, caption}
\usepackage{multirow,array}
\usepackage[switch]{lineno}
\usepackage{authblk}
\usepackage{setspace}
\usepackage{subcaption}

\newcounter{theorem}

\newtheorem{definition}{Definition}
\newtheorem{remark}{Remark}

\newtheorem{fact}{Fact}

\newtheorem{thm}[theorem]{Theorem}

\newtheorem*{conjecture*}{Conjecture}
\newtheoremstyle{nonindented}{1ex}{1ex}{}{}{\bfseries}{.}{.5em}{}
\newtheoremstyle{indented}{1ex}{1ex}{\itshape\addtolength{\leftskip}{0.6cm}\addtolength{\rightskip}{0.6cm}}{}{\bfseries}{.}{.5em}{}
\theoremstyle{nonindented}
\theoremstyle{indented}
\theoremstyle{plain}

\renewcommand{\bar}{\overline}

\def\min{\qopname\relax n{min}}
\def\max{\qopname\relax n{max}}

\def\Pr{\qopname\relax n{\mathbf{Pr}}}
\def\Ex{\qopname\relax n{\mathbf{E}}}

\newenvironment{lp*}{\begin{equation*}  \begin{array}{lll}}{\end{array}\end{equation*}}

\title{The Power of Signaling and its Intrinsic Connection to the \\ Price of Anarchy\thanks{This work is done while Nachbar is visiting the University of Virginia as a summer research intern.} 
}

\author[1]{James Nachbar}
\author[2]{Haifeng Xu}
\affil[1]{Yale University, {\tt jamie.nachbar@yale.edu}}
\affil[2]{University of Virginia, {\tt hx4ad@virginia.edu} }

\date{}

\begin{document}
\maketitle
\begin{abstract}
Strategic behaviors often render the equilibrium outcome inefficient. % --- i.e., the objective function value of an equilibrium outcome may be far from that of an optimal outcome. 
Recent  literature on information design, a.k.a. \emph{signaling}, looks to improve equilibria by selectively revealing  information to players in order to influence their actions. Most previous studies have focused on the \emph{prescriptive} question of designing optimal signaling schemes. This work departs from previous research by considering a \emph{descriptive} question, and looks to quantitatively characterize the \emph{power of signaling} (\texttt{PoS}), i.e., how much a signaling designer can improve her objective at the equilibrium outcome. 

We consider four  signaling schemes with increasing power: full information, optimal public signaling, optimal private signaling, and optimal \emph{ex-ante} private signaling. Our main result is a clean and tight characterization of the additional power each signaling scheme has over its  predecessors above in the general classes of cost-minimization and payoff-maximization games where: (1) all players minimize non-negative cost functions or maximize non-negative payoff functions; (2) the signaling designer (naturally) optimizes the sum of players' utilities. We prove that the additional power of signaling --- defined as the worst-case ratio between the equilibrium objectives of \emph{any} two signaling schemes in the above list --- is bounded precisely by the well-studied notion of the \emph{price of anarchy} (\texttt{PoA}) of the corresponding games. Moreover, we show that all these bounds are \emph{tight}.   

\end{abstract}

\section{Introduction}
A basic lesson from game theory is that strategic behaviors often render the equilibrium outcome inefficient. That is, the objective function value of an equilibrium outcome may be far from that of an optimal outcome in the absence of strategic behaviors. To reduce such inefficiency, one can ``tune'' the game equilibrium  towards a more desirable outcome, and there are two primary ways to achieve this goal: through providing \emph{incentives} or providing \emph{information}. The former approach has been widely studied in the celebrated field of mechanism design \cite{nisan2001algorithmic,procaccia2013approximate,chawla2014bayesian,conitzer2017fair}. This paper, however, focuses on the second approach, namely, improving  equilibrium by providing carefully designed information to influence players' decisions. This falls into the recent flourishing literature on  information design, a.k.a., \emph{signaling} or \emph{persuasion} \cite{dughmi2017survey,kamenica2019bayesian}. Researches in this literature so far have mainly focused on computing optimal signaling schemes for either fundamental setups \cite{Dughmi2016,Dughmi2017algorithmic,xu2020tractability,celli20} or models motivated by varied applications including auctions \cite{Emek12,Miltersen12,li2019signal}, public safety and security \cite{Xu15,Rabinovich15}, conservation \cite{Xu18},  privacy protection \cite{yan2020warn}, voting \cite{Castiglioni20}, congestion games \cite{bhaskar2016hardness,das2017reducing,Castiglioni2020SignalingIB}, recommender systems \cite{Mansour2016bayesian}, robot design \cite{keren2020information}, etc. 

Departing from the theme of all these previous works, this paper considers a different style of question. We look to characterize the \emph{power of signaling} ($\texttt{PoS}$) --- \emph{how much a signaling designer can improve her objective function of the equilibrium and can we quantitatively characterize this power?}  To our knowledge, this \emph{descriptive} question has not been formally examined before in the literature of signaling, except for a few studies which implicitly show certain  $\texttt{PoS}$-type results  in the special case of non-atomic routing games \cite{dughmi2017survey,das2017reducing,massicot2018comparative}. Nevertheless, the study of the \texttt{PoS} is extremely well-motivated. It not only deepens our understanding about signaling as an important ``knob'' to influence equilibrium, but also justifies the value of previous  prescriptive studies of optimal signaling design --- after all, the designed optimal signaling schemes are useful in practice only when the power of signaling is not negligible. Thus \texttt{PoS} is an important measure when determining the adoption of a signaling scheme in practice, especially when its tradeoff with other potential drawbacks such as communication costs \cite{gentzkow2014costly} and fairness concerns \cite{immorlica2019access} need to be balanced.

We focus on the general classes of cost-minimization games and payoff-maximization games, where each player minimizes a non-negative cost function or maximizes a non-negative payoff function. These classes of games are often studied in the  literature of the price of anarchy (\texttt{PoA}) \cite{roughgarden2015intrinsic}, and include many widely studied examples such as routing games, congestion games, most formats of auctions, valid utility games \cite{vetta2002nash}, etc. Like all standard models of signaling, players' utilities depend on a common  random state of nature $\theta$, which is drawn from a publicly known prior distribution. A signaling designer, referred to as the \emph{sender}, has an informational advantage and can access  the realized state $\theta$. The  sender is equipped with the natural objective of optimizing the total welfare, i.e., sum of players' utilities.

\vspace{3mm} 
\noindent{\bf Power of Signaling (\texttt{PoS}). } Our goal is to formally quantify the relative power of different types of signaling schemes as they become less constrained. In particular, we consider four types of signaling schemes with increasing power: full information (\texttt{FI}), optimal public signaling (\texttt{Pub}), optimal private signaling (\texttt{Pri}), and optimal \emph{ex-ante} private signaling (\texttt{exP}). \texttt{FI} is a natural \emph{benchmark} without any strategic use of information whereas \texttt{Pub}, \texttt{Pri}, \texttt{exP} are arguably the three most widely studied schemes in previous literature.\footnote{Public and private signaling has been extensively studied in previous works. Several recent works study ex-ante private signaling with motivations from recommender systems \cite{celli20,Castiglioni2020SignalingIB,xu2020tractability}. } The power of signaling of scheme \texttt{B} over \texttt{A} for \emph{any} \texttt{A} preceding \texttt{B} in  list \{\texttt{FI},  \texttt{Pub}, \texttt{Pri}, \texttt{exP}\} --- termed \texttt{PoS(B:A)} --- is defined as the ratio of the sender's utilities from scheme \texttt{A} and \texttt{B}. This ratio  is \emph{at least} $1$ for cost-minimization games and \emph{at most} $1$ for payoff-maximization games (the same as the range of the price of anarchy). Moreover, the further it is from $1$, the more powerful scheme \texttt{B} is than scheme \texttt{A}.  

\vspace{3mm}
\noindent {\bf Characterizations of \texttt{PoS}. } Our main result is a clean and tight characterization about the power of signaling.  Concretely,  for any cost-minimization game with a random state,  we prove that all the aforementioned \texttt{PoS} ratios are upper bounded by the maximum \texttt{PoA} of its corresponding realized games. Moreover,  all these upper bounds are tight in the following sense: for any ratio $r \geq 1$ and any scheme \texttt{A} preceding \texttt{B} in the list \{\texttt{FI},  \texttt{Pub}, \texttt{Pri}, \texttt{exP}\},  there exists a Bayesian cost-minimization game where all of its realized games have $\texttt{PoA} = r$ and moreover  $\texttt{PoS(B:A)} = r$ as well. We show that exactly the same results hold for payoff-maximization games --- the \texttt{PoS}s are similarly bounded by  \texttt{PoA} and all the bounds are tight.\footnote{Instead of \texttt{FI}, another natural benchmark scheme is to reveal \emph{no} information. The $\texttt{PoS}$s  compared to this benchmark turn out to be  unbounded, which we show in  Appendix \ref{append:no-info}. } 
%As applications, these results quantitatively and precisely  characterize the additional power of a less restrictive class of signaling schemes in \{\texttt{FI},  \texttt{Pub}, \texttt{Pri}, \texttt{exP}\}  in terms of guaranteeing welfare.   

Our results reveal the intrinsic connections between the power of signaling and price of anarchy. Prior to this work, it was not clear that these two concepts are inherently related ---  \texttt{PoA} characterizes the worst-case equilibrium welfare whereas \texttt{PoS} characterizes how much information can be strategically used to improve welfare. To our knowledge, recent work \cite{massicot2018comparative} is the only one to observe this connection but only for the power of public signaling over full information in non-atomic routing with afffine latency functions. Our results are   more systematic and general.  An interesting computational implication of our characterization is that the full information scheme always  serves as an $r$ approximation simultaneously for optimal public signaling, optimal private signaling and optimal ex-ante private signaling for cost-minimization or reward-maximization games, where $r$ is the worst \texttt{PoA} among realized games.   % \cite{dughmi2017survey,das2017reducing} gave concrete non-atomic routing game examples where information can help to improve welfare.  

\section{Preliminaries}
\subsection{Cost-Minimization/Payoff-Maximization Games}
%\emph{Cost-minimization games} generalize  congestion games and are commonly studied in the price of anarchy literature \cite{roughgarden2015intrinsic}. 
A cost-minimization game $G$ is a standard strategic game where each player $i$ minimizes a non-negative cost function $c_i \geq 0$. Let $n$ denotes the number of players in the game. Each player $i \in [n] = \{ 1, \cdots, n \}$  has action space $S_i$. Let $S = S_1 \times S_2 \cdots S_n$ denote the space of action profiles and $s \in S$ is a generic action profile.  
%As is standard in game theory, we allow \emph{randomized} mixed strategies. 
A (randomized) mixed strategy for player $i$ is a distribution $\mathbf{x}_i$ over $S_i$ where $x_i(s_i)$ is the probability of taking action $s_i$. By convention,  $\mathbf{x}$ denotes the profile of mixed strategies for all players, and $\mathbf{x}_{-i}$ denotes all the mixed strategies excluding $i$'s. With slight abuse of notation, let $c_i(\mathbf{x}) = \Ex_{s_i \sim \mathbf{x}_i, \forall i} c_i(s)$ denote the expected utility of player $i$ under mixed strategy $\mathbf{x}$.  There is also a \emph{global objective} which is simply to minimize the sum of the total costs $C(\mathbf{x}) = \sum_i c_i(\mathbf{x})$. 

We adopt the standard mixed-strategy \emph{Nash equilibrium} (NE) as the solution concept.  A strategy profile $\mathbf{x}^*$ is a NE if for each player $i$, $c_i(\mathbf{x}^*) \leq c_i(\mathbf{x}_i, \mathbf{x}^*_{-i})$ for any $\mathbf{x}_i \in \Delta(S_i)$. Let $X^*$ denote the set of all NEs. The well-studied concept of the price of anarchy (for mixed equilibria) for a cost-minimization game  is defined   as  follows \cite{koutsoupias1999worst,kulkarni2014robust,roughgarden2015intrinsic,feldman2016price} 
\begin{equation}\label{eq:poa}
	\texttt{POA} =  \frac{\max_{\mathbf{x}^* \in X^*}C(\mathbf{x}^*)}{\min_{s \in S}C(s)} \in [1, \infty)
\end{equation}
In other words, the $\texttt{POA}$  is the ratio between the worst Nash equilibrium and the optimal social outcome. % without any strategic behaviors. 
\begin{remark}
	The \texttt{PoA} can also be defined with respect to  pure Nash equilibrium in which case $X^*$ consists of all pure equilibria. Since not every game admits a pure Nash equilibrium, in striving for generality, we choose to analyze the version w.r.t. to mixed equilibria  since they always exist in finite games as well as in many infinite games. However, all our results --- both upper and lower bound proofs --- hold for pure equilibria as well, so long as they exist.  
\end{remark}
% \hf{need to think about whether our results generalize to any equilibrium concepts, not only pure} \jamie{If we allow mixed equilibria and have bounds on the POA for those, certainly our bounds on the power of signalling still hold. Whether our examples for the payoff-maximization case hold is another question}

\emph{Payoff-maximization games} are defined similarly; here each player $i$ maximizes expected payoff $u_i(\mathbf{x}) \geq 0$. The global objective is to maximize $U(\mathbf{x}) = \sum_i u_i(\mathbf{x})$. 
%By adapting Equation \eqref{eq:poa} to the payoff-maximization situation,
The price of anarchy here is defined similarly as $\texttt{POA} = \frac{\min_{ \mathbf{x}^* \in X^*}U(\mathbf{x}^*)}{\max_{s \in S}U(s)}$, which now lies in $[0,1]$.

This paper concerns games with uncertainty. Specifically,  players' payoffs   depend also on a \emph{random} state of nature $\theta$ drawn from support $\Theta$ with distribution $\lambda$. We use $c^{\theta}_i(s)$/$u^{\theta}_i(s)$ to denote the cost/payoff function at state $\theta$. Such a \emph{Bayesian game} is denoted by $\{ G^{\theta}\}_{ \theta \sim \lambda }$. As is standard in information design, the prior distribution $\lambda$ is publicly known to every player. We assume $\Theta$ to be finite for ease of notation, and use $\lambda(\theta)$ to denote the probability of  state $\theta$. However, all our results hold for infinite state space.

\subsection{ Signaling Schemes and Equilibrium Concepts}\label{sec:prelim:scheme}
This paper adopts the perspective of an informationally advantaged \emph{sender} who has privileged access to the \emph{realized} state $\theta$ and would like to strategically signal this information to   players in order to influence their actions. The sender is equipped with the natural objective of optimizing the sum of the players' utilities, i.e., the global objective $C(\mathbf{x})$ or $U(\mathbf{x})$, at equilibrium. We   consider  three natural types of signaling schemes with increasing generality.  

%A \textbf{private signalling scheme} is a randomized map from the possible states of nature $\Theta$ to a signal profile $\Sigma = \Sigma_1 \times \Sigma_2 \times ... \times \Sigma_n$, where there are $n$ players in the game and $\Sigma_i$ are the potential signals player $i$ can receive. We say that a signalling regime $\Sigma$ is more restrictive than $\Sigma'$ if $\Sigma \subseteq \Sigma'$. We denote the probability of selecting signal $\sigma = (\sigma_1, \sigma_2, ..., \sigma_n)$ with state of nature $\theta \in \Theta$ by $\varphi(\theta, \sigma)$. We refer to the probability of sending signal $\sigma_i$ to player $i$ given state $\theta$ as $\varphi_i(\sigma_i, \theta)$. 

\vspace{3mm}
\noindent {\bf Public Signaling.} At a high level, a \textit{public signalling scheme} constructs a random variable $\sigma$  from support $\Sigma$ --- called the \emph{signal} --- that is correlated with the state of nature $\theta$. The scheme then sends the sampled signal $\sigma$ publicly to all players, which carries information about the state $\theta$ due to their correlation. 
%Formally, a public signaling scheme  is a randomized map $\varphi$  from the set of  states of nature   $\Theta$ to a set of possible signals $\Sigma$. 
Such a public scheme $\varphi$ can be fully described by variables $\{ \varphi(\sigma; \theta) \}_{\sigma \in \Sigma, \theta \in \Theta}$ where $\varphi(\sigma; \theta)$  is the probability of sending signal $\sigma$ \emph{conditioned} on state of nature $\theta$. Adopting the standard information design assumption \cite{kamenica2011bayesian,kamenica2019bayesian}, the sender  commits to the signaling scheme before state $\theta$ is realized. Therefore, $\varphi$ is publicly known to all players. The probability of sending   signal $\sigma$ equals  $\Pr(\sigma) = \sum_\theta \lambda(\theta) \varphi(\sigma; \theta)$.  Upon receiving  signal $\sigma$, all   players perform a standard Bayesian update and infer the following posterior probability about the state $\theta$: $\Pr(\theta | \sigma) =  \lambda(\theta) \varphi( \sigma; \theta )/P(\sigma)$. 
% \vspace{-4mm}
% \begin{equation}
% \Pr(\theta | \sigma) =  \lambda(\theta) \varphi( \sigma; \theta )/P(\sigma)
% \end{equation}

Since all players receive the same information, the game will be played according to the expected cost $\bar{c}_i (s;\sigma) =  \sum_{\theta} \Pr(\theta|\sigma) c^{\theta}_i(s)$ or $\bar{u}_i (s;\sigma) =  \sum_{\theta} \Pr(\theta|\sigma) u^{\theta}_i(s)$ for all  $i$ and signal $\sigma$. We assume that players will reach a  NE of this average game. Let $C(\sigma)$   denote the sender's expected cost at equilibrium under signal $\sigma$ and $C(\varphi) = \sum_{\sigma} \Pr(\sigma) C(\sigma) $ denote the expected sender cost under signaling scheme $\varphi$. Like the \texttt{PoA} literature, when there are multiple Nash equilibria, we always adopt the \emph{worst}  one in  our analysis. 
Notations for payoff maximization are defined similarly. %These notations apply to the following private and ex-ante private signaling as well.  

% Let $S^*_{pub}(\sigma;\varphi)$ denote the set of all PNEs of the game induced by public signal $\sigma$ under scheme $\varphi$. When the scheme $\varphi$ is clear from the context, we omit $\varphi$ and simply use $S^*_{pub}(\sigma)$ to denote the set.   

% Upon sending some signal $\sigma$, each player then knows the posterior probability of each state, and can formulate their strategies based on the expected value of the game state, as well as the expected beliefs of the other players. The player then selects some strategy $s$ which is optimal with respect to their beliefs about the state of the game. We can view sending a private signalling scheme, then, as being equivalent to recommending $s$, which is optimal with the information the player is given.       

\vspace{3mm}
\noindent {\bf Private Signaling.} Private signaling relaxes public signaling by allowing the sender to send different, and possibly correlated, signals to different players. Specifically, let $\Sigma_i$ denote  the set of possible signals to player $i$ and $\Sigma = \Sigma_1  \times ... \times \Sigma_n$ denote the set of all possible signal profiles. With slight abuse of notation, a private signaling scheme can   be similarly captured by variables $\{ \varphi(\sigma;\theta) \}_{\theta \in \Theta, \sigma \in \Sigma}$. When signal profile $\sigma$ is restricted to have the same signal to all players, this degenerates to public signaling.  Private signaling leads to a truly Bayesian game where each player holds different information about the state of nature. The standard solution concept in this case is the \emph{Bayes correlated equilibrium} (BCE) introduced by Bergemann and Morris \cite{bergemann2016bayes}, which consists of all outcomes that can possibly arise at Bayes Nash equilibrium under all possible signaling schemes. Standard revelation-principle type argument shows that   signals of private signaling schemes in a BCE can be interpreted as \emph{obedient action recommendations} \cite{kamenica2011bayesian,bergemann2016bayes,Dughmi2016}. That is, $\Sigma_i$ can W.L.O.G. be $S_i$ and $\Sigma = S$. An action recommendation $s_i$ to player $i$ is \emph{obedient} if following this recommended action is indeed a best response for $i$, or formally, for any $ s_i, s'_i \in S_i$ we have 
\begin{equation}\label{eq:obedience}
	\sum_{\theta \in \Theta, s_{-i} \in S_{-i}} \varphi(s_i, s_{-i};\theta) \lambda(\theta) c_i^{\theta}(s_i, s_{-i}) \geq \sum_{\theta \in \Theta, s_{-i} \in S_{-i}} \varphi(s_i, s_{-i};\theta) \lambda(\theta) c_i^{\theta}(s'_i, s_{-i})    
\end{equation}

\vspace{2mm}
\noindent  \textbf{Ex-Ante Private Signaling.}  Motivated by recommender system applications, recent works   \cite{xu2020tractability,celli20,Castiglioni2020SignalingIB} relax the obedience constraints \eqref{eq:obedience} of BCE to a \emph{coarse correlated equilibrium} type of obedience constraints, described as follows:
\begin{equation}\label{eq:cce-obedience}
	\sum_{\theta \in \Theta, s  \in S} \varphi(s_i, s_{-i};\theta) \lambda(\theta) c_i^{\theta}(s_i, s_{-i}) \geq \sum_{\theta \in \Theta, s \in S } \varphi(s_i, s_{-i};\theta) \lambda(\theta) c_i^{\theta}(s'_i, s_{-i}), \forall  s'_i \in S_i.  
\end{equation}
That is, for any player $i$, following the recommendation is better than opting out of the signaling scheme  and acting just according to his prior belief.  A signaling scheme satisfying Constraint \eqref{eq:cce-obedience} is dubbed an \emph{ex-ante} private  scheme \cite{celli20,Castiglioni2020SignalingIB}. % Note that,  Constraint \eqref{eq:cce-obedience} can be viewed as the sum of Constraint \eqref{eq:obedience} over $s_i$ for a given $s'_i$. Therefore, ex-ante private signaling  is a relaxation of private signaling. 

\section{The Power of Signaling (\texttt{PoS})}\label{sec:PoS}
% With the previous description of signaling schemes of increasing generality, we are now ready to 
We now formalize the \emph{Power of Signaling} (\texttt{PoS}) in cost-minimization and payoff-maximization games.  
Intuitively, the  \texttt{PoS} characterizes how much additional power  a  class of signaling schemes has over another. Formally, let $\Phi^a$ and $\Phi^b$ be two classes of signaling schemes (e.g.,  public and private  schemes).  We say $\Phi^b$ is \emph{less restricted} than $\Phi^a$,  conveniently denoted as $\Phi^a \subseteq \Phi^b$,  if $\varphi \in \Phi^b$ whenever  $\varphi \in \Phi^a$. 

\begin{definition}[\texttt{PoS} of $\Phi^b$ over $\Phi^a$]\label{def:pos}
	For any two classes of signaling schemes  $\Phi^a, \Phi^b$ where $\Phi^b$ is less restricted than $\Phi^a$ (i.e., $ \Phi^a \subseteq \Phi^b$), the power of signaling of $\Phi^b$ over $\Phi^a$, or $ \texttt{PoS}(\Phi^b : \Phi^a)$ for short, is defined as  
	\begin{equation*}
		%\begin{split}
		\texttt{PoS}(\Phi^b: \Phi^a) =  \frac{\min_{\varphi \in \Phi^a} C(\varphi)}{\min_{\varphi \in \Phi^b} C(\varphi)}    \bigg( or \, \frac{\max_{\varphi \in \Phi^a} U(\varphi)}{\max_{\varphi \in \Phi^b} U(\varphi)}  \bigg), 
		%\end{split}
	\end{equation*}
	for cost-minimization (or payoff-maximization) games. 
\end{definition}    
In other words, \texttt{PoS}  is the ratio between the objectives of the optimal   scheme from signaling class $\Phi^a$ and that from a less restricted class $\Phi^b$. 
Similar to  the \texttt{PoA} ratio,   $\texttt{PoS}(\Phi^b: \Phi^a) \geq 1$ for cost-minimization games, and the larger this ratio is, the more powerful $\Phi^b$ is over $\Phi^a$.  In contrast,  $\texttt{PoS}(\Phi^b: \Phi^a) \leq 1$ for payoff-maximization games, and the smaller this ratio is, the more powerful $\Phi^b$ is over $\Phi^a$. If both the numerator and denominator are 0, we say the \texttt{PoS} is 1; if only the denominator is 0,  the \texttt{PoS} is $+\infty$.   

Though \texttt{PoS} is well-defined for any two classes of signaling schemes, in this paper we primarily consider the following well-studied  classes of signaling schemes: 

\begin{itemize}
	\item $\Phi^1$ or \texttt{FI}:  full information ; 
	\item  $\Phi^2$ or \texttt{Pub}: public signaling schemes; 
	\item  $\Phi^3$ or \texttt{Pri}: private signaling schemes; 
	\item  $\Phi^4$ or \texttt{exP}: ex-ante private signaling schemes. 
\end{itemize}  

%The equilibrium concepts for each class are as described in Section \ref{sec:prelim:scheme}, which are all standard in the literature. 
The full information class \texttt{FI} only contains a single signaling scheme, i.e., fully revealing the state $\theta$. This serves as a benchmark scheme where information is not strategically signaled. Another natural benchmark scheme is to reveal \emph{no} information. We  show in   Appendix \ref{append:no-info} that the $\texttt{PoS}$  compared to this benchmark turns out to be  unbounded. 

\section{\texttt{PoS} in Cost-Minimization Games}\label{sec:cost}

% \hf{Two questions: (1) can the lower bound proof generalize to mixed Nash? (2) For the tightness proofs, are there any mixed Nash equilibrium that is even worse than your worst pure nash? }

The main result of this section is the following tight characterization about the \texttt{PoS} ratios in cost-minimization games. 
\begin{thm}\label{thm:cost}
	For any Bayesian  cost-minimization game $\{ G^{\theta} \}_{\theta \sim \lambda}$, let $\texttt{PoA}_{\max} = \max_{\theta} \texttt{PoA}(G^{\theta})$ denote the worst  \texttt{PoA} ratio among  game $G^{\theta}$s. We have 
	\begin{equation}\label{eq:pos-upper}
		\texttt{PoS}(\Phi^j:\Phi^i) \leq \texttt{PoA}_{\max}, \quad \forall \,  1\leq i < j \leq 4. 
	\end{equation}  
	Moreover, these upper bounds are all tight in the following sense: for any $r \geq 1$ and   $1\leq i < j \leq 4 $, there exits a Bayesian cost-minimization game $\{ G^{\theta} \}_{\theta \sim \lambda}$  with $\texttt{PoA}(G^{\theta}) = r$ for any $ \theta$ and $\texttt{PoS}(\Phi^j:\Phi^i) = r$ as well.  
	% \jamie{In addition, since public signals are a subset of private ones, and private signals are a subset of ex-ante private, our bound is tight for Full:Private, Full:Ex-Ante Private, and Public:Ex-Ante Private.}
\end{thm}
\noindent
The remainder of this section is devoted to the proof of Theorem \ref{thm:cost}. The following simple observation follows from Definition \ref{def:pos} of the \texttt{PoS}, and will be useful for proving the tightness of the   bounds in Inequality \eqref{eq:pos-upper}.  
\begin{fact}\label{fact1}
	For any $1 \leq i < j < j' \leq 4$, we have  $$\texttt{PoS}(\Phi^j:\Phi^i) \leq \texttt{PoS}(\Phi^{j'}:\Phi^i).$$ 
\end{fact}

As a consequence of Fact \ref{fact1}, if we prove the tightness of $\texttt{PoS}(\Phi^2:\Phi^1)$, i.e., $\texttt{PoS(Pub:FI)}$ by constructing an example with $\texttt{PoS(Pub:FI)}=\texttt{PoA}_{\max}$, the example must satisfy  $\texttt{PoS(Pri:FI)} = \texttt{PoS(exP:FI)} =\texttt{PoA}_{\max}$ as well, implying their tightness. Therefore, to prove the tightness of Inequality \eqref{eq:pos-upper}, we only need to prove the tightness of \texttt{PoS(Pub:FI)}, \texttt{PoS(Pri:Pub)} and \texttt{PoS(exP:Pri)}. 

\subsection{A Simultaneous Proof of all the \texttt{PoS} Upper Bounds}

We first prove all  the upper bounds in Inequality \eqref{eq:pos-upper} through a unified result, summarized in the following theorem.   
\begin{thm}\label{thm:cost-bound}
	For any Bayesian  cost-minimization game $\{ G^{\theta} \}_{\theta \sim \lambda}$, we  have  $ \texttt{PoS}(\Phi^b:\Phi^a) \leq \max_{\theta} \texttt{PoA}(G^{\theta})$  for any two classes of signaling schemes $\Phi^a , \Phi^b$ satisfying $\Phi^a \subseteq \Phi^b$ and that the full information scheme is contained in $\Phi^a$.
\end{thm}
\begin{proof} 
	%Consider some cost-minimization Bayesian game $G$, with state of nature $\theta$ drawn from $\Theta$, inducing $G_\theta$ with probability $\lambda(\theta)$. In addition, suppose the price of anarchy is defined for each $G_\theta$ and is bounded by $r \geq 1$.
	
	Let $\texttt{PoA}_{\max} = \max_{\theta} \texttt{PoA}(G^{\theta})$ be the worst  (i.e., the maximum) price of anarchy ratio among game $G^{\theta}$s.  
	%Denote by $C(G_\theta)$ the cost of the worst equilibrium of the game $G_\theta$, which we know exists since the price of anarchy is defined. 
	Denote by $C^*(G_\theta) = \min_{s\in S} C(S)$  the minimum total social cost among all outcomes (not necessarily an equilibrium) for   game $G_\theta$; let $s^{\theta*}$ be an strategy profile that achieves  $C^*(G_\theta)$. $\varphi^0$ denotes the full information revelation scheme. 
	
	Observe that for any signaling scheme $\varphi$ we must have $C(\varphi) \geq   \sum_{\theta} \lambda(\theta) C^*(G^{\theta}) $ because regardless of how players act in the scheme $\varphi$, its expected total cost can never be less than the minimum possible total cost $\sum_{\theta} \lambda(\theta) C^*(G^{\theta})$.  Now since $\Phi^a$ contains $\varphi^0$, we thus have  
	\begin{eqnarray*}
		\min_{\varphi' \in \Phi^a} C(\varphi') &\leq& C(\varphi^0) \\
		&=&  \sum_{\theta \in \Theta} \lambda(\theta) C(G_\theta) \\ 
		&\leq& \sum_{\theta \in \Theta} \lambda(\theta)\cdot  r C^*(G_\theta) \\
		& \leq & r C(\varphi), \text{ for any scheme }\varphi  
	\end{eqnarray*}
	where $C(G_\theta)$ is the worst (i.e., maximum) equilibrium cost of game $G^{\theta}$ and the second inequality is by the definition of the price of anarchy. As a result, $\texttt{PoS}(\Phi^b:\Phi^a) = \frac{\min_{\varphi' \in \Phi^a} C(\varphi')}{ \min_{\varphi \in \Phi^b} C(\varphi) } \leq r$, as desired.
	% \jamie{Changed to $\Phi^a$ and $\Phi^b$, since we already define $\Phi^1$ and $\Phi^2$ earlier}
	% \hf{Can change to $\Phi^a$ and $\Phi^b$ to be consistent with Definition 1}
\end{proof}

\subsection{Tightness of the Upper-Bound for \texttt{PoS(Pub:FI)}}  

{\bf Non-Atomic Routing.} 
It turns out that all the \texttt{PoS} bounds in Theorem \ref{thm:cost}  are tight in a special and well-studied class of cost-minimization games, i.e., non-atomic routing.  The game takes place on a directed graph $G=(V,E)$ with $V$ as the vertex set and $E$ as the edge set.  There is a continuum of players, each controlling a negligible amount of flow characterized by a pair of nodes $(s,t)$ where $s\in V$ is the starting node of the flow and $t\in V$ is its destination.  %Let $f_{(s , t)}$ denote the aggregated flow demand  from $s$ to $t$.  
In non-atomic routing with incomplete information, each edge $e\in E$ can be described by a congestion function $c^{\theta}_e(x)$ which depends on the total amount of flow $x$ on   edge $e$ as well as a random state of nature $\theta \in \Theta$.  Each player $(s,t)$  optimizes her own utility by taking a minimum-cost directed path from $s$ to $t$. The sender minimizes overall congestion cost. There is an  essentially unique pure Nash equilibrium for non-atomic routing under a public scheme. Therefore, equilibrium selection is not an issue in non-atomic routing.

We now show the tightness of \texttt{PoS(Pub:FI)} for any ratio $r \geq 1$ via a non-atomic routing game example, which implies the tightness of  \texttt{PoS(Pri:FI)}  and  \texttt{PoS(exP:FI)}   by Fact \ref{fact1}.  Consider a variant of   Pigou's example \cite{pigou2013economics},  as depicted in Figure \ref{fig:cost-tight1},  where cost functions are described on each edge. The traffic demand from $s_1$ to $t$ is set as $d(\alpha) = \left(\frac{1}{\alpha+1} \right)^\frac{1}{\alpha}$ whereas the demand from $s_2$ to $t$ is $1 -d(\alpha)$.  Each state of nature occurs with probability $0.5$. 

\begin{figure}[h]
	\captionsetup[subfigure]{labelformat=empty}
	\scalebox{1}{
		\begin{subfigure}{.55\textwidth}
			\centering
			\begin{tikzpicture}
			\SetVertexMath
			\Vertex[x=0, y=0, L=s_2]{s_2}
			\Vertex[x=3, y=1, L=t]{t}
			\Vertex[x = 0, y=2, L=s_1]{s_1}
			
			\tikzstyle{EdgeStyle}=[pre]
			\Edges[label={$x^\alpha$}](t, s_1)    
			\tikzstyle{EdgeStyle}=[pre]
			\Edges[label={$1$}](t, s_2)
			\tikzstyle{EdgeStyle}=[pre, bend right = 40]
			\Edges[label={$2$}](s_1, s_2)
			\tikzstyle{EdgeStyle}=[pre, bend left = 40]
			\Edges[label={$0$}](s_1, s_2)
			
			\end{tikzpicture}
			\captionof{figure}{$\theta_1$}
	\end{subfigure}}
	\scalebox{1}{
		\begin{subfigure}{.3\textwidth}
			\centering
			\begin{tikzpicture}
			\SetVertexMath
			\Vertex[x=0, y=0, L=s_2]{s_2}
			\Vertex[x=3, y=1, L=t]{t}
			\Vertex[x = 0, y=2, L=s_1]{s_1}
			
			\tikzstyle{EdgeStyle}=[pre]
			\Edges[label={$x^\alpha$}](t, s_1)    
			\tikzstyle{EdgeStyle}=[pre]
			\Edges[label={$1$}](t, s_2)
			\tikzstyle{EdgeStyle}=[pre, bend right = 40]
			\Edges[label={$0$}](s_1, s_2)
			\tikzstyle{EdgeStyle}=[pre, bend left = 40]
			\Edges[label={$2$}](s_1, s_2)
			\end{tikzpicture}
			\captionof{figure}{$\theta_2$}
	\end{subfigure}}
	\caption{\vspace{-2mm} A Tight Example for \texttt{PoS(Pub:FI)} \label{fig:cost-tight1}}
\end{figure}

First, we compute the price of anarchy (\texttt{PoA}) of each game, as a function of $\alpha$. Clearly, at equilibrium no flow will pass through the edge with cost 2 since deviating to the edge with cost 0 is strictly better. It is easy to see that at equilibrium all flow will go through the edge with cost $x^\alpha$, leading to total congestion $1$ at equilibrium. The optimal flow, however, is that all flow at $s_1$ goes through edge $(s_1, t) $  and all flow at $s_2$ goes through $(s_2, t)$, leading to minimum total congestion   $\left(\frac{1}{\alpha+1}\right)^{\frac{\left(\alpha+1\right)}{\alpha}}+1\ -\ \left(\frac{1}{\alpha+1}\right)^{\frac{1}{\alpha}}$. Therefore, the \texttt{PoA} as a function of $\alpha$ in this instance is 
\begin{equation}\label{eq:flow-poa}
	\texttt{PoA} = \frac{1}{\left(\frac{1}{\alpha+1}\right)^{\frac{\left(\alpha+1\right)}{\alpha}}+1\ -\ \left(\frac{1}{\alpha+1}\right)^{\frac{1}{\alpha}}}, 
\end{equation} which is a continuous function of $\alpha > 0$. Standard analysis shows that this function tends to $\infty$ as $\alpha \to \infty$ and tends to $1$ as  $\alpha \to 0^+$.  For the special case of $\alpha=0$, it can be directly verified that the \texttt{PoA} ratio is  1. Therefore, this \texttt{PoA} ratio can take any value $r\geq 1$ with a proper choice of $\alpha$.

We now consider the cost of the full information scheme \texttt{FI}. In this case, all flow will always go through the edge with cost $1$ at equilibrium, leading to total cost $1$. The optimal public scheme in this example happens to be revealing no information. Without being able to distinguish the zero-cost edge from the edge of cost $2$, all flow at $s_2$ will take the $(s_2,t)$ path. This achieves the minimum total cost, rendering the \texttt{PoS(Pub:FI)}  ratio equal the \texttt{PoA} for any $n$.  
% leads to the optimal flow. We see that if we reveal no information to either player, the expected costs of both edges between $s_2$ and $s_1$ becomes 1. The expected cost of the traffic starting at $s_2$ going to $t$ through $s_1$ therefore becomes greater than 1, which means the new equilibrium is for all the traffic at $s_2$ to take the bottom edge, with cost 1. We see that this scheme is optimal, by the same proof we used to derive the price of anarchy for this game. Therefore, the power of signalling is again $\frac{1}{\left(\frac{1}{n+1}\right)^{\frac{\left(n+1\right)}{n}}+1\ -\ \left(\frac{1}{n+1}\right)^{\frac{1}{n}}}$, which is equal to the Price of Anarchy, which we have seen is continuous, equal to 1 at $n = 0$, and diverges as $n \to \infty$.

\subsection{Tightness of the Upper-Bound for \texttt{PoS(Pri:Pub)}}\label{sec:cost:pri} 
We now show the tightness of \texttt{PoS(Pri:Pub)}  for any ratio $r \geq 1$, which implies the tightness of \texttt{PoS(exP:Pub)}  by Fact \ref{fact1}. We construct a Bayesian non-atomic routing game as depicted in Figure \ref{fig:cost-tight2}, which can be viewed as another variant of Pigou's example.\footnote{This example generalizes an earlier example observed by Cheng, Dughmi and Xu \cite{dughmi2017survey}. It also avoided their use of the (unrealistic) $\infty$ flow cost, and relies on  a less trivial analysis due to the finite edge cost.} There is a $1$ unit of flow demand from  $s$ to $t$. Each state has equal probability $0.5$. 
\begin{figure}[h]
	\captionsetup[subfigure]{labelformat=empty}
	\scalebox{1}{
		\begin{subfigure}{.55\textwidth}
			\centering
			\begin{tikzpicture}
			\SetVertexMath
			\Vertex[x=0, y=0, L=s]{s}
			\Vertex[x=3, y=0, L=t]{t}
			
			\tikzstyle{EdgeStyle}=[pre, bend right=40]
			\Edges[label={2}](t, s)
			\tikzstyle{EdgeStyle}=[pre, bend left=40]
			\Edges[label={$1$}](t, s)
			\tikzstyle{EdgeStyle}=[pre]
			\Edges[label = $x^\alpha$](t,s)
			
			\end{tikzpicture}
			\captionof{figure}{$\theta_1$}
	\end{subfigure}}
	\scalebox{1}{
		\begin{subfigure}{.3\textwidth}
			\centering
			\begin{tikzpicture}
			\SetVertexMath
			\Vertex[x=0, y=0, L=s]{s}
			\Vertex[x=3, y=0, L=t]{t}
			
			\tikzstyle{EdgeStyle}=[pre, bend right=40]
			\Edges[label={$x^\alpha$}](t, s)
			\tikzstyle{EdgeStyle}=[pre, bend left=40]
			\Edges[label={$1$}](t, s)
			\tikzstyle{EdgeStyle}=[pre]
			\Edges[label = 2](t,s)
			\end{tikzpicture}
			\captionof{figure}{$\theta_2$}
	\end{subfigure}}
	\caption{A Tight Example for \texttt{PoS(Pri:Pub)} \label{fig:cost-tight2}}
\end{figure}

Similarly to the calculation for the Example in Figure  \ref{fig:cost-tight1},  the \texttt{PoA} for each game here also equals that as described in Equation \eqref{eq:flow-poa}. We now argue that the expected cost of any public signaling scheme will equal $1$ in this example --- i.e., all public schemes are equally bad and will not be able to reduce any congestion.  Any public signal gives the same information about the game state to all players. Let $\lambda \in [0,1]$ denote the posterior probability of $\theta_1$ given any public signal. W.l.o.g., consider the case $\lambda \geq  0.5$ since the other case is symmetric. The top edge will have expected cost $2\lambda   +  x^\alpha(1-\lambda) > 1$ for any $x>0$, therefore this edge will never be taken since the bottom edge is a strictly better choice. Consequently, players will be choosing between the middle edge, with expected cost function $f(x) =:  x^\alpha \lambda + 2(1-\lambda)$, and the bottom edge with fixed cost $1$. Note that $f(0) = 2(1-\lambda)\leq 1$ and   $f(1) = 1-\lambda \geq 1$. Therefore, at the unique equilibrium,  the amount of flow through the middle edge will be exactly the $x^*$ such that $f(x^*) = 1$ whereas the remaining flow will be through the bottom edge. The expected total cost at this equilibrium is $1$.      

% We now consider the players being partially sure of which edge is the edge with cost $x^n$ If the players know this edge with probability $.5 \leq \lambda \leq 1$, we see the expected value of the edge which is more likely to have cost two is greater than 1, which means no players will ever take it. However, the edge which is more likely to have cost $x^n$ has expected cost $\lambda x^n + 2(1-\lambda)$, which we see is less than 1 initially for $0.5 < \lambda$. However, we see that since we can send all our traffic through this edge and arrive at cost $2 - \lambda$, which is always greater than or equal to 1. Therefore, at equilibrium, the expected cost of this edge to each unit of traffic must be 1. If it was greater than 1, traffic would divert to the certain edge, which has cost 1, and if it was less than 1, players would divert traffic from the constant edge to it. We have already shown the the expected value of this edge is always greater than 1 when all traffic is going through it, so we know that if it has cost less than 1, it is always possible to divert traffic to it. The remaining traffic then takes the constant edge, showing that the cost of the optimal public scheme is 1, which occurs at full-revelation. 

Finally, we show that the optimal private signaling  will be able to induce the optimal flow, concluding our tightness proof. Consider the following private signaling scheme: revealing full information to a randomly selected   $x^* = \left(1/(\alpha+1)\right)^{1/\alpha}$ fraction of the players, and revealing no information to the remaining players. The $x^*$ fraction of players given the full information has a   dominant action of taking the edge with cost $x^\alpha$ at the told state. For the remaining players with no information, their cost of taking either the top or the middle edge will be at least $2 \times (1/2) + (x^*)^\alpha \times (1/2) > 1$. Therefore, their optimal response will be taking the bottom edge. This leads to exactly the optimal flow for each state, as desired.

\subsection{Tightness of the Upper-Bound for \texttt{PoS(exP:Pri)}}

Finally, we prove the tightness of $\texttt{PoS(exP:Pri)}$.  Consider the non-atomic routing game depicted in Figure \ref{fig:cost-tight3}.\footnote{This example generalizes an earlier example observed by Cheng \cite{Cheng-example} after removing (unrealistic) $\infty$ cost edges. Note that the analysis of the example becomes less obvious without the use of $\infty$ costs.}  There is one unit of flow from $s$ to $t$ and the two states $\theta_1, \theta_2$ occurs with equal probability $0.5$.  
\begin{figure}[ht]
	\captionsetup[subfigure]{labelformat=empty}
	\scalebox{1}{
		\begin{subfigure}{.55\textwidth}
			\centering
			\begin{tikzpicture}
			\SetVertexMath
			\Vertex[x=0, y=0, L=s]{s}
			\Vertex[x=3, y=0, L=t]{t}
			
			\tikzstyle{EdgeStyle}=[pre, bend right=20]
			\Edges[label={$\alpha+1$}](t, s)
			\tikzstyle{EdgeStyle}=[pre, bend right=50]
			\Edges[label={$\alpha+1$}](t, s)
			\tikzstyle{EdgeStyle}=[pre, bend left=20]
			\Edges[label={$1$}](t, s)
			\tikzstyle{EdgeStyle}=[pre, bend left = 50]
			\Edges[label = $x^\alpha$](t,s)
			
			\end{tikzpicture}
			\captionof{figure}{$\theta_1$}
	\end{subfigure}}
	\scalebox{1}{
		\begin{subfigure}{.3\textwidth}
			\centering
			\begin{tikzpicture}
			\SetVertexMath
			\Vertex[x=0, y=0, L=s]{s}
			\Vertex[x=3, y=0, L=t]{t}
			
			\tikzstyle{EdgeStyle}=[pre, bend left=20]
			\Edges[label={$\alpha+1$}](t, s)
			\tikzstyle{EdgeStyle}=[pre, bend left=50]
			\Edges[label={$\alpha+1$}](t, s)
			\tikzstyle{EdgeStyle}=[pre, bend right=20]
			\Edges[label={$1$}](t, s)
			\tikzstyle{EdgeStyle}=[pre, bend right = 50]
			\Edges[label = $x^\alpha$](t,s)
			\end{tikzpicture}
			\captionof{figure}{$\theta_2$}
	\end{subfigure}}
	\caption{A Tight Example for \texttt{PoS(exP:Pri)} \label{fig:cost-tight3}}
%	\vspace{-3mm}
\end{figure}

Similar to the analysis for previous examples, the \texttt{PoA} for each game equals also the function described in Equation \eqref{eq:flow-poa}, which takes value in $[1, \infty)$ as we vary the parameter $\alpha$. 

% With each state having a $0.5$ chance of being selected. We see that for each of these games, the equilibrium flow is again all traffic going through the edge with cost function $x^n$, since this is always less than or equal to 1. In the optimal flow, once again no traffic goes through an edge with cost $n+1$, so we are again choosing how much traffic goes through the edge with cost $x^n$ with the rest going through the edge with cost 1. Therefore, the cost of the optimal flow is the same as in the previous example: $\left(\frac{1}{n+1}\right)^{\frac{\left(n+1\right)}{n}}+1\ - \left(\frac{1}{n+1}\right)^{\frac{1}{n}}$. The price of anarchy is therefore again $\text{POA} = \frac{1}{\left(\frac{1}{n+1}\right)^{\frac{\left(n+1\right)}{n}}+1\ -\ \left(\frac{1}{n+1}\right)^{\frac{1}{n}}}$, which is a continuous function defined for all $n > 0$ which diverges as $n \to \infty$.  At $n = 0$, the optimal private scheme has the same cost as the optimal flow, resulting in price of anarchy and power of signalling equal to 1. 

Next we argue that the optimal private signaling scheme is full-revelation, with cost 1. Recall from the preliminary section, any private scheme can be viewed as obedient action recommendations. Numbering the edges from the top to the bottom as edge $1,2,3,4$, we claim that any obedient action recommendation should never recommend edge $2$ and $3$. This is because if a non-zero amount of players are recommended, e.g., to edge $2$,  switching to edge $1$ or $4$  will be strictly better. In particular, if the players are certain that they are at state $\theta_1$, they will prefer to switch to edge $4$. Otherwise, there is non-zero probability that they are at state $\theta_2$. In this case, switching to edge $1$ is strictly better. Consequently, the optimal private scheme can be captured by two variables: (1) $x$: the amount of flow recommended to edge $4$ at state $\theta_1$ (thus edge $1$ consumes the remaining $1-x$ amount); (2) $y$:  the amount of flow recommended to edge $1$ at $\theta_2$. It can be shown that the parameterized total expected cost  $[x^{\alpha+1} + (\alpha+1)(1-x) + y^{\alpha+1} + (\alpha+1)(1-y)]/2$  is minimized at $x = 1, y=1$, i..e, the full information scheme.  
%  --- the amount of flow there is strictly less than $1$ since edge $2$ has non-zero amount of flow. 

%Consequently, the obedient flow recommendation of any private scheme would only recommend edge $1$ and $4$. Given this observation, we  argue that the optimal private scheme should never recommend any players to take a path with cost $\alpha+1$. In particular, given any private signaling scheme, let $x$ be the amount of flow recommended to edge $4$ at state $\theta_1$ (thus edge $1$ is recommended $1-x$ amount) and $y$ be the amount of flow recommended to edge $1$ at state $\theta_2$. The total expected cost of this private signaling scheme is then $f(x,y) = [x^{\alpha+1} + (\alpha+1)(1-x) + y^{\alpha+1} + (\alpha+1)(1-y)]/2$. The partial derivative of $f$ on $x$ equals $(\alpha+1)x^\alpha - (\alpha+1) \leq 0$. Therefore it is minimized at $x = 1, y=1$, i..e, the full information scheme.    

%Our argument so far shows that the optimal private signaling scheme achieves total cost $1$. 
Finally, we show that the optimal ex-ante private scheme achieves the minimum possible total social cost, %$\left(1/(\alpha+1)\right)^{\frac{\left(\alpha+1\right)}{\alpha}}+1\ -\ \left(1/(\alpha+1)\right)^{\frac{1}{\alpha}}$, 
concluding the tightness proof of \texttt{PoS(exP:Pri)}. In particular, consider the ex-ante private scheme that induces the optimal flow by recommending a randomly chosen $\left(1/(\alpha+1)\right)^{1/\alpha}$ amount of the flow to the edge with cost $x^\alpha$ at any state and the remaining flow to the edge with cost $1$. This is indeed obedient in the ex-ante sense because opting out from this scheme and taking any path will lead to cost at least $ (\alpha+1 + \frac{1}{ \alpha+1} )/2$,  which is at least $1$ and thus  is larger than its expected utility  in the scheme (less than $1$).

\section{\texttt{PoS} in Payoff-Maximization Games}

In this section, we show that a similar tight characterization as in Theorem \ref{thm:cost} holds for payoff-maximization games as well.

\begin{thm}\label{thm:payoff}
	For any Bayesian payoff-maximization game $\{ G^{\theta} \}_{\theta \sim \lambda}$, let $\texttt{PoA}_{\min} = \min_{\theta} \texttt{PoA}(G^{\theta})$ denote the worst-case \texttt{PoA} ratio among  game $G^{\theta}$s. We have
	\begin{equation}\label{eq:pos-lower}
		\texttt{PoS}(\Phi^j:\Phi^i) \geq \texttt{PoA}_{\min}, \quad \forall 1\leq i < j \leq 4. 
	\end{equation}  
	Moreover, these lower bounds are all tight in the following sense: for any $r \in (0, 1]$ and any $1\leq i < j \leq 4 $, there exits a Bayesian payoff-maximization  game $\{ G^{\theta} \}_{\theta \sim \lambda}$  with $\texttt{PoA}(G^{\theta}) = r$ for any $ \theta$ and $\texttt{PoS}(\Phi^j:\Phi^i) = r$ as well.  
	% \jamie{In addition, since public signals are a subset of private ones, and private signals are a subset of ex-ante private, our bound is tight for Full:Private, Full:Ex-Ante Private, and Public:Ex-Ante Private.}
\end{thm} 

%\todo{need rephrase, maybe use $\Phi^1, \Phi^2, \Phi^3$ instead. }

\begin{remark}
	In payoff-maximization games, the smaller \texttt{PoS} is, the more powerful signaling is. Therefore, Inequality \eqref{eq:pos-lower} is a lower bound for the  \texttt{PoS} ratio but an upper bound for the power of signaling. Similar discrepancies also arise in the definition of the \texttt{PoA} for payoff-maximization game \cite{roughgarden2015intrinsic}. 
\end{remark} 

The remaining of this section is devoted to the proof of Theorem \ref{thm:payoff}. A simultaneous proof of all the \texttt{PoS} lower bounds in Inequality \eqref{eq:pos-lower} follow an analogous argument as that for Theorem \ref{thm:cost-bound}, and thus is omitted here. We only prove their tightness. One might wonder whether the tightness proof here can be adapted from that for cost-minimization games by simply reversing minimizing cost functions to be maximizing their negations (plus a large constant to make it positive). The answer turns out to be \emph{no}. We illustrate the detailed reasons in the Appendix \ref{append:reverse-route}, but at a high level there are at least two reasons. First, the optimal flow for congestion minimization may not be optimal any more in the negation of the game. Second, some \texttt{PoA} ratios cannot be achieved in the negation of routing games.    % --- we illustrate the reason at the end of this section. 

Our tightness proof here  requires   carefully constructed  payoff-maximization games and analysis. Thanks to Fact \ref{fact1}, we will only need to prove the tightness of \texttt{PoS(Pub:FI)}, \texttt{PoS(Pri:Pub)} and \texttt{PoS(exP:Pri)}.

\subsection{Tightness of the Lower-Bound for \texttt{PoS(Pub:FI)}  }
Consider the following game  played by two players P1, P2  where $\alpha, \epsilon$ are parameters satisfying $\alpha - 1 >  2\epsilon > 0$. Each player has two actions, conveniently denoted as $A, B$. There are two states $\theta_1, \theta_2$ with equal probability $0.5$. The only  difference of the two states is the payoffs for action profile $(B,A)$, which is $(1,\alpha)$ at state $\theta_1$ and $(\alpha,1)$ at state $\theta_2$.  % \hf{need to make the following table cells to be centered better} \jamie{Attempted to display both states in one table (let me know if this isn't readable)} 
\begin{figure}[h]
	\centering
	\begin{subfigure}{.4\textwidth}
		\centering
		\begin{tabular}{cc|c|c|c|}
			& \multicolumn{1}{c}{} & \multicolumn{2}{c}{P2}\\
			& \multicolumn{1}{c}{} & \multicolumn{1}{c}{$A$}  & \multicolumn{1}{c}{$B$} \\ \cline{3-4}
			\multirow{2}*{P1} & $A$ & $(1+\epsilon,1)$ & $(1-\epsilon,1-\epsilon)$\\
			\cline{3-4} & $B$ & $(1,\alpha)$ & $(1,1+\epsilon)$\\ \cline{3-4}
		\end{tabular}
		\caption{\qquad \qquad \quad $\theta_1$ \vspace{3mm}}
	\end{subfigure}
	\qquad 
	\begin{subfigure}{.45\textwidth}{
			\centering
			\begin{tabular}{cc|c|c|c|}
				& \multicolumn{1}{c}{} & \multicolumn{2}{c}{P2}\\
				& \multicolumn{1}{c}{} & \multicolumn{1}{c}{$A$}  & \multicolumn{1}{c}{$B$} \\ \cline{3-4}
				\multirow{2}*{P1} & $A$ & $(1+\epsilon,1)$ & $(1-\epsilon,1-\epsilon)$\\
				\cline{3-4} & $B$ & $(\alpha, 1)$ & $(1,1+\epsilon)$\\ \cline{3-4}
		\end{tabular}}
		\caption{ \qquad $\theta_2$ \vspace{3mm}}
	\end{subfigure}
\end{figure}

We first consider full information (\texttt{FI}). At state $\theta_1$,  action $A$ is a strictly dominant action for Player 2. This implies that both players choosing action $A$ is the only  Nash equilibrium, resulting in sender utility $2+\epsilon$. However, the optimal outcome at state $\theta_1$ is that Player 1 chooses $B$ and Player 2 chooses $A$, leading to total payoff $\alpha+1$ ($> 2+2\epsilon$ since $\alpha-1 > 2 \epsilon$). State $\theta_2$ is symmetric. The expected utility of full information  is  $2 + \epsilon$, and the \texttt{PoA} of each game is $\frac{2+\epsilon}{\alpha+1}$.  

We now show that the optimal public signaling scheme can maximize total payoff. Consider the  scheme which reveals no information at all. The only change is that the expected payoffs of action profile $(B,A)$ becomes $\frac{\alpha+1}{2}$ for both players. We see then, that $A$ remains a strictly dominant strategy for Player 2 and  $B$ becomes a strictly dominant action for Player 1. The  only equilibrium  is the action profile $(B,A)$, resulting in  expected sender utility $\alpha + 1$. Therefore, $\texttt{PoS(Pub:FI)}=\frac{2+ \epsilon}{\alpha+1}$, equaling the price of anarchy. %  Since $\frac{2+2\epsilon}{C+1} \leq \texttt{PoS(exP:FI)}  \leq \texttt{PoS(Pri:FI)} \leq \texttt{PoS(Pub:FI)} = \frac{2+2\epsilon}{C+1} $, we thus know that they are all tight in this constructed example. We see that 
If $\alpha \to 1$, the \texttt{PoS} ratio tends to 1, and  $\alpha \to \infty$ with a fixed $\epsilon$ makes the \texttt{PoS} ratio continuously tend to 0. Setting $\alpha = 1+\epsilon$ gives a trivial case where the \texttt{PoS} and \texttt{PoA} are both 1. Thus, the ratio achieves all possible values within $(0, 1]$.  

\subsection{Tightness of the Lower-Bound for \texttt{PoS(Pri:Pub)}}

% {\bf The Robber's Game.}  As a means of illustrating our constructed game, consider two  robbers robbing a bank. They have triggered the alarm, and are pressed for time. As they enter the vault of the bank, they have to make a decision. In the vault are two \emph{safes}, and a huge pile of cash. Each of them has three options: attempt to crack Safe 1 (action $S_1$), Safe 2 (action $S_2$), or simply take as much cash as they can carry (action $C$), worth $\epsilon$. In order to protect the bank's valuable items, one of the safes is a decoy safe, and is \emph{empty}. Inside the other safe are two objects: a gold bar worth $\epsilon + \delta$, where $1 \geq \epsilon > \delta>0$ and $\delta$ is an extremely small constant ,  and an extremely delicate but valuable crystal worth $1$. The two robbers (players) have very different skill sets. Player 1 (P1) is a lock picking expert, and Player 2  (P2) is a demolitions expert. The players' payoffs equal  whatever they individually steal from the vault. Since P2 is a demolitions expert, only he is capable of carrying the crystal without breaking it. If P1 cracks the safe using his lockpicking, P2 can take the crystal on the way out, leaving P1 with the gold bar. However, if P2 cracks the safe, the crystal will be destroyed due to his bad lock picking skill, leaving him only the gold bar. In addition, if both players try to crack the same safe, they get in each others' way and are forced to leave with nothing. 

{\bf The Robber's Game.}  As a means of illustrating our constructed game, consider two robbers robbing a bank. They have triggered the alarm, and are pressed for time. As they enter the vault, they have to make a decision. In the vault are two \emph{safes}, and a huge pile of cash. There are three options: attempt to crack Safe 1 (action $S_1$), Safe 2 (action $S_2$), or simply take as much cash as they can carry (action $C$), worth $\alpha$. In order to protect the bank's valuable items, one of the safes is a decoy safe, and is \emph{empty}. Inside the other safe are two objects: a gold bar worth $\alpha + \epsilon$, where $1 \geq  \alpha > \epsilon > 0$, and an extremely delicate but valuable crystal worth $1$. The two robbers (players) have very different skill sets. Player 1 (P1) is a lock picking expert, and Player 2  (P2) is a demolitions expert. The players' payoffs equal whatever they individually steal from the vault. Only P2 is capable of carrying the crystal without breaking it. If P1 cracks the safe using his lockpicking, P2 can take the crystal on the way out, leaving P1 with the gold bar. However, if P2 cracks the safe, the crystal will be destroyed due to his explosions, leaving him only the gold bar. In addition, if both players try to crack the same safe, they get in each others' way and are forced to leave with nothing. 

\begin{figure}[h]
	\centering
	\scalebox{.9}{
		\begin{subfigure}{.4\textwidth}
			\begin{tabular}{cc|c|c|c|}
				& \multicolumn{1}{c}{} & \multicolumn{3}{c}{P$2$} \\
				& \multicolumn{1}{c}{} & \multicolumn{1}{c}{$S_1$}  & \multicolumn{1}{c}{$C$}  & \multicolumn{1}{c}{$S_2$} \\\cline{3-5}
				& $S_1$ & $(0,0)$ & $(0,\alpha)$ & $(0,\alpha + \epsilon)$ \\ \cline{3-5}
				P$1$  & $C$ & $(\alpha,0)$ & $(\alpha,\alpha)$ & $(\alpha,\alpha + \epsilon)$ \\\cline{3-5}
				& $S_2$ & $(\alpha + \epsilon,1)$ & $(\alpha + \epsilon,1)$ & $(0,0)$ \\\cline{3-5}
			\end{tabular}
			\caption{\qquad \qquad \quad $\theta_1$ \vspace{3mm}}
	\end{subfigure}}
	\qquad \qquad \quad
	\scalebox{.9}{
		\begin{subfigure}{.5\textwidth}
			
			\begin{tabular}{cc|c|c|c|}
				& \multicolumn{1}{c}{} & \multicolumn{3}{c}{P$2$} \\
				& \multicolumn{1}{c}{} & \multicolumn{1}{c}{$S_1$}  & \multicolumn{1}{c}{$C$}  & \multicolumn{1}{c}{$S_2$} \\\cline{3-5}
				& $S_1$ & $(0,0)$ & $(\alpha + \epsilon,1)$ & $(\alpha + \epsilon,1)$ \\ \cline{3-5}
				P$1$  & $C$ & $(\alpha, \alpha + \epsilon)$ & $(\alpha,\alpha)$ & $(\alpha, 0)$ \\\cline{3-5}
				& $S_2$ & $(0,\alpha + \epsilon)$ & $(0,\alpha)$ & $(0,0)$ \\\cline{3-5}
				
			\end{tabular}
			\caption{\qquad \qquad \quad $\theta_2$ \vspace{3mm}}
	\end{subfigure}}
\vspace{-3mm}
	\caption{Payoffs of the robber's game achieving tight $\texttt{PoS(Pri:Pub)}$; each state occurs with probability $0.5$. \label{fig:payoff-pri}}
\end{figure}

Concretely, the payoff matrix for the aforementioned game is in the above tables. At $\theta_1$, safe 1 is empty, and at $\theta_2$, safe 2 is empty.  Note that $\theta_2$ simply exchanges the payoffs of strategies $S_1$ and $S_2$ symmetrically for both players. A-priori, both robbers do not know the state and share the common uniform random prior. The sender is a heist leader and knows which safe is which. The sender gets to pocket a cut of the total haul, so naturally she is interested in maximizing the total utility. In public signaling,   both robbers use the same radio to communicate with the sender, but in private signaling each player has his own communication radio.

% \hf{ Questions about the robbers' game motivation --- $(S_2, S_2)$ gives utility $(0,0)$ in $\theta_1$, this means if both cracks the safe, it will explode still despite that P1 is good at lock picking? Also, why both taking action $C$, they will each get $\epsilon$? From a motivating perspective, should they split that $\epsilon$ since there is only one pile of money there? } \jamie{My explanation is that if they both try to crack the safe they get in each other's way and nothing gets done. For the cash, I want to say that there is a huge amount of cash, but they can only carry a certain amount - $\epsilon$ worth}

We first  calculate the \texttt{PoA} for the game at each state. W.l.o.g., we consider   equilibria for state $\theta_1$ as $\theta_2$ is symmetric. Since action $C$ dominates action $S_1$ (i.e., cracking the empty safe) for both players, without loss of generality we will assume both players do not play action $S_1$ in our following analysis. It is easy to see that $(S_2, C)$ and $(C,S_2))$ are the only two pure Nash equilibria after excluding action $S_1$. We now consider mixed equilibrium in this game. Note that since both players have only two actions,  both players must randomize in any mixed equilibrium.  Let $\lambda_1, \lambda_2 \in (0,1)$ denote the probability that player 1 and 2 play $C$, respectively. We have  $\lambda_2 (\alpha + \epsilon) = \alpha$ since player 1's both actions must be equally good,  and similarly $\lambda_1 \alpha + (1-\lambda_1) = (\alpha+\epsilon) \lambda_1$. This implies $\lambda_1 = 1/(1+\epsilon)$ and $\lambda_2 = \alpha/(\alpha+\epsilon)$. The total expected payoff of this mixed strategy equilibrium is $\alpha + (\alpha + \epsilon)/(1+\epsilon)$, which is the smallest equilibrium total payoff.  The largest social payoff is however $\alpha +\epsilon +1$. Therefore, the \texttt{PoA} of this game is   $\frac{2\alpha + \epsilon + \alpha \epsilon}{(1 + \alpha + \epsilon)(1+\epsilon)}$. 

% Suppose Player 1 plays C with probability $\lambda_1$ and $S_2$ with probability $1 - \lambda_1$, and Player 2 plays $C$ with probability $\lambda_2$ and $S_2$ with probability $1 - \lambda_2$. We see the expected payout for Player 1 is $\epsilon \lambda_1 + \lambda_2(1 - \lambda_1)(\epsilon + \delta)$. We see that if Player 1 was to switch to playing only $C$, they would earn $\epsilon$. We see that if $\lambda_2 < 1$, as $\delta \to 0$, $\epsilon > \lambda_2(\epsilon + \delta)$, implying that if $\lambda_1 < 1$, $(1 - \lambda_1) \epsilon > \lambda_2(1 - 
% \lambda_2)(\epsilon + \delta)$, which means that $\epsilon > \lambda_1\epsilon + \lambda_2(1 - 
% \lambda_2)(\epsilon + \delta)$. Therefore, at equilibrium, as $\delta \to 0$, for $\lambda_2 < 1$, $\lambda_1 \to 1$. This means the only equilibrium solutions are at $\lambda_2 = 1$, or $\lambda_1 = 1$. We see that if $\lambda_2 = 1$, Player 1 always plays $S_2$, resulting in utilities $(\epsilon + \delta, 1)$, and if $\lambda_1 = 1$, Player 2 always plays $S_2$, resulting in utilities $(\epsilon, \epsilon + \delta)$. The optimal case is clearly the first of these equilibria, and the worst equilibrium is the second, so the Price of Anarchy of this game is $\frac{2\epsilon + \delta}{1 + \epsilon + \delta}$. 
% \vspace{2mm}

We now show that the optimal public signalling scheme is to reveal full information, assuming worst-case equilibrium selection. We prove this by arguing that for any public signal with posterior probability $p \in [0,1]$ of state $\theta_1$, the sender's utility is at most $U_0 = \alpha + (\alpha + \epsilon)/(1+\epsilon)$ in the worst equilibrium, which is the sender utility of full information revelation. %Note that $U_0 >  \alpha + (\alpha + \epsilon \alpha)/(1+\epsilon) = 2\alpha$. 
%Consequently, the optimal sender utility must be revealing full information, leading to utility $U_0$. 
This follows a case analysis, depending on whether $p$ is between $\frac{\epsilon}{\alpha + \epsilon} $ and $\frac{\alpha}{\alpha + \epsilon}$, greater than or equal to $\frac{\alpha}{\alpha + \epsilon}$, or less than or equal to $\frac{\epsilon}{\alpha + \epsilon}$. By cases: 
\begin{enumerate}
	\item When $p (\alpha + \epsilon) < \alpha $ and $(\alpha + \epsilon)(1-p) < \alpha$, i.e.,  $\frac{\epsilon}{\alpha + \epsilon} < p< \frac{\alpha}{\alpha + \epsilon}$ (recall that our parameter choice satisfies $\alpha > \epsilon)$. In this case, for P1,  the utility $\alpha$ of the safer action $C$  is strictly larger than the best possible expected utility $(\alpha + \epsilon)(1-p)$ of taking   action $S_1$ and larger than the best possible expected utility $(\alpha + \epsilon)p$ of taking action $S_2$ as well. Given that P1 will always take $C$, P2 strictly prefers $C$ as well as his utility  $(\alpha + \epsilon)(1-p)$ for $S_1$ and utility $(\alpha + \epsilon)p$ for $S_2$ are both smaller.  
	%the additional reward $\epsilon$ of getting the gold bar is outweighed by uncertainty in the safe and taking the ``safer'' action $C$ is a strictly dominant strategy for P1. 
	Therefore, both players taking action $C$ is the unique   equilibrium, leading to sender utility $2\alpha$ which is less than $ U_0$ since $U_0 >  \alpha + (\alpha + \epsilon \alpha)/(1+\epsilon) = 2\alpha$. 
	\item When $p (\alpha + \epsilon) \geq  \alpha $. In this case, state $\theta_1$ is very likely and action $S_1$ is strictly \emph{dominated} by action $C$ for P1 and thus will  not be taken by P1. We show that there exists a mixed strategy that has sender utility worse than $U_0$. In particular, consider the following mixed strategies: (1) P1 chooses action $C$ with probability $p_C = \frac{p+\alpha(1-p)}{\epsilon p+p}$ and action $S_2$ with remaining probability $1-p_C$; (2) P2 chooses action $C$ with probability $q_C = \frac{\alpha}{(\alpha+\epsilon)p}$ and action $S_2$ with remaining probability $1-q_C$. We claim that this is a mixed strategy equilibrium. In particular, both action $C$ and $S_2$ have expected utility $\alpha$ to P1 and both action $C$ and $S_2$ have expected utility $p_C(\alpha+\epsilon)p$. Therefore, the sender's utility at the worst mixed equilibrium is at most  \begin{align*}
		\alpha + \frac{p+\alpha(1-p)}{\epsilon p+p} (\alpha+\epsilon)p &= \alpha +  \frac{\alpha+\epsilon}{\epsilon +1}[p+\alpha(1-p)]\\ &\leq \alpha +  \frac{\alpha+\epsilon}{\epsilon +1} = U_0.
	\end{align*}
	\item When $(1-p)(\alpha + \epsilon) \geq  \alpha $, i.e. $p \leq \frac{\epsilon}{\alpha+\epsilon}$. This case is symmetric to Case 2, and thus has the same conclusion. 
\end{enumerate}

Finally, we argue that the optimal private scheme results in the optimal social outcome. Consider the private scheme that reveals no information to Player 2 but full information to Player $1$. Given no information, Player 2 has a strictly dominant action $C$ since $\epsilon < \alpha$. With full information, Player 1 will always open the non-empty safe.  This leads to the optimal outcome and sender utility $1+\alpha + \epsilon$.  Therefore, the $\texttt{PoS(Pri:Pub)}$ for this game equals the  $\texttt{PoA} =\frac{2\alpha + \epsilon + \alpha \epsilon}{(1 + \alpha + \epsilon)(1+\epsilon)}$. We see that when $\alpha = 1$ and $\epsilon \to 0$, the \texttt{PoA} tends to 1; when $\alpha \to 0$ and $\epsilon \to 0$, the \texttt{PoA} tends to 0. When $\alpha = 1$ and $\epsilon = 0$, we have a trivial case, with \texttt{PoA} and \texttt{PoS} 1. Thus its ratio takes any value within $(0,1]$.

\subsection{Tightness of the Lower-Bound for \texttt{PoS(exP:Pri)}}
Finally, for any $r \in (0, 1]$, we show the tightness of \texttt{PoS(exP:Pri)}. Consider the same two robbers, robbing the same bank. However, the bank has now upgraded its anti-theft countermeasures. Instead of a decoy safe, there is now an entire decoy vault. Naturally, any robber who goes into it will leave empty-handed. In the real vault, there is again a large pile of cash, but now (only) one safe. A robber can choose to take cash for a guaranteed payoff ($\alpha + \epsilon$ for Player 1, $\alpha$ for Player 2). The safe contains the same fragile, valuable crystal as in the previous section, but this time, no gold bar. If both players attempt to crack the safe, they get in each others' way and will fail. After cracking the safe, each player has enough time to take cash instead of the contents of the safe if they choose. We assume  $1 \geq \alpha > \epsilon \geq 0$ and $\alpha + \epsilon < 1$.  % If they crack the safe, they only have enough time to take $\epsilon$ worth. 
If Player 2 attempts to take the cash, and Player 1 cracks the safe, Player 2 can take the crystal on the way out, \emph{instead of} the cash. However, if Player 2 does not go into the correct vault, he will leave with nothing. The detailed payoff is as in the following table. Here, $\theta_1$ corresponds to the first vault being the decoy one, whereas $\theta_2$ corresponds to the second vault being empty. Each state occurs with equal probability $0.5$. The sender as the heist leader knows which vault is empty. 

\begin{figure}[h]
	\centering
	\begin{subfigure}{\textwidth}
		\centering
		\begin{tabular}{cc|c|c|c|c|c|}
			& \multicolumn{1}{c}{} & \multicolumn{4}{c}{P2}\\
			& \multicolumn{1}{c}{} & \multicolumn{1}{c}{$C_1$}  & \multicolumn{1}{c}{$S_1$} & \multicolumn{1}{c}{$C_2$} & \multicolumn{1}{c}{$S_2$}\\\cline{3-6}
			\multirow{4}*{P1} & $C_1$ & $(0,0)$ & $(0,0)$ & $(0,\epsilon)$ & $(0, \epsilon)$\\
			\cline{3-6} & $S_1$ & $(0,0)$ & $(0,0)$ & $(0,\alpha)$ & $(0, \alpha)$\\
			\cline{3-6} & $C_2$ & $(\alpha + \epsilon,0)$ & $(\alpha + \epsilon,0)$ & $(\alpha + \epsilon,\alpha)$ & $(\alpha + \epsilon,\alpha)$\\ \cline{3-6} & $S_2$ & $(\alpha,0)$ & $(\alpha,0)$ & $(\alpha,1)$ & $(0, 0)$ \\ \cline{3-6}
		\end{tabular}
		\caption{\qquad \qquad \quad $\theta_1$ \vspace{3mm}}
	\end{subfigure}
	\begin{subfigure}{\textwidth}
		\centering
		\begin{tabular}{cc|c|c|c|c|c|}
			& \multicolumn{1}{c}{} & \multicolumn{4}{c}{P2}\\
			& \multicolumn{1}{c}{} & \multicolumn{1}{c}{$C_1$}  & \multicolumn{1}{c}{$S_1$} & \multicolumn{1}{c}{$C_2$} & \multicolumn{1}{c}{$S_2$}\\\cline{3-6}
			\multirow{4}*{P1} & $C_1$ & $(\alpha + \epsilon,\alpha)$ & $(\alpha + \epsilon,\alpha)$ & $(\alpha + \epsilon, 0)$ & $(\alpha + \epsilon,0)$\\
			\cline{3-6} & $S_1$ & $(\alpha,1)$ & $(0,0)$ & $(\alpha, 0)$ & $(\alpha, 0)$\\
			\cline{3-6} & $C_2$ & $(0, \alpha)$ & $(0,\alpha)$ & $(0,0)$ & $(0,0)$\\ \cline{3-6} & $S_2$ & $(0, \alpha)$ & $(0,\alpha)$ & $(0,0)$ & $(0, 0)$ \\ \cline{3-6}
		\end{tabular}
		\caption{\qquad \qquad \quad $\theta_2$ \vspace{3mm}}
	\end{subfigure}
	\vspace{-3mm}
	\caption{ Payoffs of the robber's game variant with tight $\texttt{PoS(exP:Pri)}$; each state occurs with  probability $0.5$.}
\end{figure}
% \todo{Add motivation for the delta change}

We first calculate the \texttt{PoA} of each game. For the game at state $\theta_1$, it is easy to see that action $C_2$ strictly dominates all other actions for Player 1. This leads to $(C_2, C_2)$ and $(C_2, S_2)$ be the two unique NEs and the total payoff of any  equilibrium   is $(\alpha + \epsilon) + \alpha = 2 \alpha + \epsilon$. The maximum total payoff however is $\alpha + 1$, under action profile $(S_2, C_2)$. Similar analysis holds for state $\theta_2$. The \texttt{PoA} for each game is thus $(2\alpha + \epsilon)/(1+\alpha)$ (recall $\alpha + \epsilon<1$ in our game). 

% $C_1, S_1$ for both players. We thus only need to consider action $C_2, S_2$. It is easy  Calculating the price of anarchy, we first find the equilibria of this game. W.l.o.g. for $\theta_1$, we see that Player 1 choosing $C_2$ is a dominant strategy. Therefore, Player 2 either plays $C_2$ or $S_2$, the result is identical. Either of these cases is therefore the worst equilibrium, with total utility $2\epsilon + \delta$ The optimal outcome is clearly Player 1 playing $S_2$ and Player 2 playing $C_2$, which has total utility $1 + \epsilon$. and the Price of Anarchy of this game is $\frac{2\epsilon}{1 + \epsilon}$.

We now argue that an optimal private  scheme is to reveal full information. Crucially, Player 1's strictly dominant action is to take $C_1$ or $C_2$, whichever is more likely to get the $\alpha + \epsilon$ cash amount in the posterior distribution of his private signal. Given this, Player 2's optimal action is  to choose $C_1/S_1$ or $C_2/S_2$, whichever is more likely to get the $\alpha$ payoff (from cash or from opening the safe)  in the posterior distribution of  his private signal. Consequently, any partial information will lead to Player 1 utility at most $\alpha + \epsilon$ and Player 2 utility at most $\alpha$. This renders full information revelation optimal, leading to total player utility $2\alpha + \epsilon$. %  Consider some arbitrary signalling scheme, where Player 1 knows the true state with probability $\lambda_1 \geq 0.5$ and Player 2 knows the true state with probability $\lambda_2 \geq 0.5$. We see that a dominant strategy for Player 1 is always choosing either $C_1$ or $C_2$, whichever is more likely to give a payoff. Given that this is the case, Player 2's best response is selecting the same option. Neither player has any incentive to deviate, since deviating to either the cash or the safe that is less likely would provide less than or equal value, and deviating to the safe that is more likely would provide the same value as they get trying to take the cash. This means that the value of the worst equilibrium of this game is $(\lambda_1 + \lambda_2)\epsilon$, which is maximized when $\lambda_1 = \lambda_2 = 1$. Therefore, the optimal private scheme is full-revelation, with utility $2\epsilon$.

Finally, we show that an ex-ante signaling scheme induces the optimal outcome. The scheme simply recommends $(S_2, C_2)$ at state $\theta_1$ and $(S_1, C_1)$ at state $\theta_2$. This satisfies the ex-ante obedience constraint \eqref{eq:cce-obedience} because: (1) if Player 1 opts out and acts according to his prior belief, he gets expected utility at most $\frac{1}{2}(\alpha + \epsilon)$, which is strictly less than his utility  $\alpha$ in the scheme; (2) Player 2 gets utility 1 in the scheme and   certainly does not want to opt out.  %    If we recommend the correct safe to Player 1, and the correct pile of cash to Player 2, we see that we induce the optimal outcome, which has total value $1 + \epsilon$. In addition, this scheme is ex-ante persuasive, since this scheme results in the maximum reward for both players, given that the utility of choosing a random state is at most half of the utility they receive from the current recommendation. 
Therefore, the  \texttt{PoS(exP:Pri)} ratio in this game equals precisely the \texttt{PoA} of each game $\frac{2\alpha + \epsilon}{1 + \alpha}$. This ratio tends to  1 as $\alpha + \epsilon \to  1$ and tends to 0 as $\alpha, \epsilon \to 0^+$. If $\alpha + \epsilon = 1$, the \texttt{PoA} and \texttt{PoS(exP:Pri)} are trivially equal to 1. Therefore, the ratio takes any value $r \in (0, 1]$.

\section{Discussions and Future Work}
In this paper, we initiate and formalize the concept of the power of signaling (\texttt{PoS}). In the general classes of cost-minimization and payoff-maximization games, we show that the \texttt{PoS} is inherently related to, in fact precisely characterized by, the \emph{price of anarchy} (\texttt{PoA}). 

There are many possibilities for future research. In our analysis, we use the full information scheme (\texttt{FI}) as the benchmark, since the no information scheme (\texttt{NI}) will lead to infinite bound. However, another natural and stronger benchmark is the \emph{better} of these two schemes \texttt{FI, NI}. With such a stronger benchmark, the power of signaling will not increase. One interesting question is whether we will have strictly less  power of signaling when compared to this strong benchmark. We observe that positive answers to this characterization question will have interesting algorithmic implications. For example, even restricting to non-atomic routing games with linear latency functions, if one can show that $\texttt{PoS(Pub:max\{FI,NI\})} = r$ for some $r< 4/3$, this would imply that the better between  \texttt{FI} and \texttt{NI} --- which certainly can be computed efficiently --- will serve as an $r$-approximation for optimal public signaling. However, it is proved in \cite{bhaskar2016hardness} that it is NP-hard to approximate optimal public signaling for this setting to be within a ratio better than $4/3$ in this setting. This shows that    for  non-atomic routing with linear latency functions, even $\texttt{PoS(Pub:max\{FI,NI\})}$ is strictly smaller than the \texttt{PoA} ratio $4/3$, it will be NP-hard to prove this conclusion since any proof implies an efficient approximation algorithm with ratio strictly better than $4/3$.  

The above discussion considers stronger benchmark schemes. Another direction is to study the power of signaling for more restricted classes of  signaling schemes, such as schemes with limited communication power \cite{dughmi2016persuasion} or schemes with costly communication \cite{gentzkow2014costly}? For these classes of schemes, the power of signaling  will also decrease. It is interesting to understand how much  these restrictions limit the power of signaling. On the other hand, the sender's objective considered in this paper is the total social welfare. In many applications of signaling, the sender's objective may be different from the welfare, e.g., revenue as in auctions. In this case, tools beyond the price of anarchy may be needed since \texttt{PoA} mainly concerns welfare.  It is an intriguing open direction to characterize the power of signaling in these settings.

\bibliographystyle{plain}
\bibliography{refer}

\clearpage

\appendix

\setstretch{1.15}
\section{\texttt{PoS} w.r.t the No Information (\texttt{NI}) Benchmark}\label{append:no-info}
In the main body of our paper, we choose full information (\texttt{FI}) as our benchmark scheme. One might wonder what happens if the no information scheme is used instead. It turns out that no information (\texttt{NI}) may lead to very bad social welfare, and examples are fairly easy to construct.  Massicot and Langbort gave a Bayesian cost-minimization game with $\texttt{PoA}=4/3$ for each game --- more concretely, a  non-atomic routing game example with affine latency function --- such that $\texttt{PoS(Pub:NI)} \to \infty$  \cite{massicot2018comparative}.

We now exhibit a simple reward-maximization game with $\texttt{PoA} = 1$ but $\texttt{PoS(Pub:NI)} \to 0$.  
Consider a (trivial) game with  $n$ actions $A_1, A_2, A_3, ..., A_n$ and a single player. There are $n$ equally likely states of nature, with each state of nature $\theta_i$ gives utility $1$ to action $A_i$ and  $0$ utility to all other actions. The price of anarchy of this game is trivially 1, since there is only one reward-maximizing player, and full information as the optimal public scheme achieves  optimal welfare $1$. However,  in the case of no information, the player can only get expected utility $\frac{1}{n}$. As $n\to \infty$, $\texttt{PoS(Pub:NI)} \to 0$.

\section{Non-tightness of \texttt{PoS} in ``Reverse'' Routing}\label{append:reverse-route}
When proving the tightness for \emph{payoff-maximization} games, a very natural first attempt is, perhaps, to convert the previously constructed cost-minimization routing games 
%of Section \ref{sec:cost}  
into  payoff-maximization  games by flipping the  sign of cost functions and adding a large constant to make it positive. One example by reversing our game constructed for the tightness of \texttt{PoS(Pri:Pub)}  is depicted in Figure \ref{fig:reverse-routing}. That is, any edge with cost function $c_e(x)$ in our original construction can be changed to instead having payoff $N - c_e(x)$ for large positive constant $N$. 
Clearly, this is a valid payoff-maximization game, which we term ``reverse'' routing game for convenience. 
\begin{figure}[ht] 
	\centering
	\begin{tikzpicture}
	\SetVertexMath
	\Vertex[x=0, y=0, L=s]{s}
	\Vertex[x=3, y=0, L=t]{t}
	
	\tikzstyle{EdgeStyle}=[pre, bend right=40]
	\Edges[label={N - 2}](t, s)
	\tikzstyle{EdgeStyle}=[pre, bend left=40]
	\Edges[label={$N-1$}](t, s)
	\tikzstyle{EdgeStyle}=[pre]
	\Edges[label = $N-x^\alpha$](t,s)
	\end{tikzpicture}
	\captionof{figure}{A reverse routing  example \label{fig:reverse-routing}}
\end{figure}
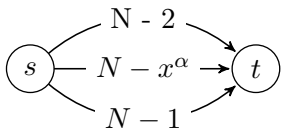

We argue this natural adaptation of our previous routing game constructions  in this  way does not produce an example with, e.g.,  tight $\texttt{PoS(Pri:Pub)}$ ratio. This is why we must turn to new constructions of payoff-maximization games. There are two reasons.  First, this adaption cannot lead to any price of anarchy ratio within $(0,1)$. In particular, it can be verified that the \texttt{PoA} of the game in Figure \ref{fig:reverse-routing} is at least $1/2$ since $N \geq 2$. The second major reason is that optimal routes in the standard cost-minimization routing game may not be optimal any more in its natural adaption to the reward-maximization situation. For example, in cost minimization, one never wants to route through a cycle but in its reward maximization variant, we would like to route through a cycle as much as possible to collect rewards (such examples are fairly easy to construct).

\end{document}